\documentclass[letterpaper,10pt,conference]{ieeeconf}

\IEEEoverridecommandlockouts
\usepackage{amsmath,amssymb,mathtools}
\usepackage{bm}
\usepackage{graphicx}
\usepackage{epstopdf}
\usepackage{psfrag,color}
\usepackage[hidelinks]{hyperref}

\DeclarePairedDelimiter\norm{\lVert}{\rVert}

\newtheorem{assumption}{Assumption}
\newtheorem{theorem}{Theorem}
\newtheorem{lemma}{Lemma}
\newtheorem{corollary}{Corollary}
\newtheorem{remark}{Remark}

\newcommand{\R}{\mathbb{R}}

\newcommand{\xhat}{\hat{x}}

\newcommand{\tilx}{\tilde{x}}

\newcommand{\res}{r}
\newcommand{\Id}{\mathbb{I}}

\title{\LARGE \bf State Observers with Prescribed Output Error Bounds}

\author{Nilay~Kant%
\thanks{The author is with the Department of Mechanical and Aerospace Engineering, Missouri University of Science and Technology, Rolla, MO, USA.}%
\thanks{Corresponding author: Nilay Kant (email: nilaykant@mst.edu)}%
}

\begin{document}

\maketitle

\begin{abstract}
This paper presents a state observer design for continuous linear time-invariant (LTI) systems subject to unknown bounded disturbances that enforces a prescribed bound on the observer output error (residual). The proposed observer augments a Luenberger observer with state resets triggered when the residual norm reaches a prescribed bound. The reset map guarantees contraction of the residual at jump instants, forward invariance of the prescribed residual envelope, and strict decrease of the estimation-error Lyapunov function at every reset, while preserving the uniform boundedness properties of a standard Luenberger observer. Notably, the observer design does not require knowledge of the disturbance bound. Simulation results confirm the analysis: the residual remains within the prescribed bound, whereas a standard Luenberger observer with the same gains violates this bound.
\end{abstract}
	

%

\section{Introduction}\label{sec1}

State observers are essential in modern control and monitoring systems, enabling output-feedback control, fault detection \cite{patton1997observer, jeong2019fault}, and supervisory decision making \cite{sneider2002observer} when full state measurements are unavailable. Classical continuous-time observers provide asymptotic convergence and robustness guarantees under suitable assumptions. However, these do not explicitly regulate transient mismatches between measured and estimated outputs, referred to as the observer residual. In observer-based control, large residuals may degrade performance when state estimates are used for feedback. In such applications, it may therefore be desirable to directly regulate the residual and ensure that it remains within a prescribed bound. This motivates an observer design that can enforce a designer-specified bound on the residual.

\emph{Novelty of the work:} This paper introduces an observer architecture for LTI systems subject to bounded disturbances with unknown bounds that enforces a prescribed bound on the observer residual. The design augments a continuous-time Luenberger observer with state resets triggered when the residual norm reaches the prescribed bound. A closed-form reset law contracts the residual at each jump, thereby rendering the prescribed residual envelope forward invariant. Despite the induced state jumps, the estimation-error Lyapunov function strictly decreases at every reset, while the uniform boundedness properties of the Luenberger observer are preserved.

The closest related work \cite{torben2023resetting} employs residual-based resets to improve transient estimation performance without high continuous gains; however, the residual threshold serves only as a triggering mechanism. In contrast, the present work treats the residual bound as an invariant performance specification and enforces it by design. High-gain observers \cite{khalil2017high,khalil2014high} can also reduce transient residual deviations, but may exhibit peaking and increased sensitivity to measurement noise. The proposed architecture instead regulates transient residual behavior through discrete resets without requiring aggressive continuous gains or knowledge of the disturbance bound. Thus, unlike standard robustness bounds that depend on the observer dynamics and disturbance magnitude, the proposed approach \emph{directly enforces a designer-prescribed residual bound}.

Hybrid and reset mechanisms \cite{haddad2006impulsive, goebel2012hybrid} in observer design are not new, but have been studied in different contexts. For example, reset strategies have been used to improve transient response or reduce performance indices such as the $L_2$ gain of the estimation error \cite{hosseini2019lmi}. Resetting adaptive observers have been developed for joint state and parameter estimation in uncertain nonlinear systems \cite{paesa2011reset, paesa2011optimal}, and reset observers for linear time-varying delay systems have been analyzed to establish global asymptotic and finite-gain $L_2$ stability \cite{zhao2014reset}. Event-triggered observers have also been proposed to regulate communication load by transmitting measurements only when a triggering condition is satisfied \cite{petri2021event}. The remainder of the paper is organized as follows. Section~II introduces notation. Section~III presents the system model, assumptions, and problem formulation. Section~IV develops the reset observer and its theoretical guarantees. Section~V presents simulation results, followed by concluding remarks in Section~VI.

\section{Notation}

For a matrix $X$, $\norm{X}$ denotes the induced $2$-norm. For a vector $v$, $\norm{v}$ denotes the Euclidean norm.
$P \succ 0$ denotes that $P$ is positive definite. For a symmetric matrix $\mathcal{S}$, $\lambda_{\min}(\mathcal{S})$ and $\lambda_{\max}(\mathcal{S})$ denote its minimum and maximum eigenvalues. For a hybrid signal $v(t)$, $v(t_k^-) \triangleq v^-$ and $v(t_k^+) \triangleq v^+$ denote its left and right limits at a jump time $t_k$. Finally, $\mathbb{I}$ denotes the identity matrix of appropriate dimension.

\section{Problem Formulation}

\subsection{System Dynamics}
Consider the continuous-time LTI system:
\begin{subequations}\label{eq:plant}
\begin{align}
\dot{x}(t) &= A x(t) + B u(t) + d(t) \label{eq:plant_x}\\
y(t) &= C x(t) \label{eq:plant_y}
\end{align}
\end{subequations}
where $x\in\R^n$ is the state vector, $u\in\R^m$ is bounded and piecewise continuous control input, and $y\in\R^p$ is the system output. The system matrices $
A\in\R^{n\times n},
B\in\R^{n\times m},
C\in\R^{p\times n}
$ are known, and the signal $d(t) \in\R^{n}$ represents an unknown exogenous disturbance acting on the system. Let $t_0$ denote the initial time.

\begin{assumption}\label{ass:rate}
The disturbance $d(t)$ is Lebesgue measurable and essentially bounded, with
\begin{equation}
\|d(t)\| \leq \bar d, \qquad \text{for almost all } t \geq t_0 \label{eq:dist_rate}
\end{equation}
where $\bar d > 0$ is a finite scalar and need not be known.
\end{assumption}

\begin{assumption}\label{ass:detect}
The pair $(A,C)$ is detectable.
\end{assumption}

\begin{assumption}\label{ass:init_set}
The estimation error satisfies $\|\tilx(t_0)\| \le \bar x_0$ for some $\bar x_0>0$.
\end{assumption}

\begin{assumption}\label{ass:fullrankC}
The matrix $C\in\R^{p\times n}$ has full row rank.
\end{assumption}

\subsection{Continuous-Time Observer}

A Luenberger observer for the states is given by:
\begin{equation}
\dot{\xhat}(t) = A \xhat(t) + B u(t) + L\left[y(t) - C \xhat(t) \right] \label{eq:ct_observer}
\end{equation}
where $L \in \R^{n\times p}$ is the observer gain. Define the estimation error
$\tilx \triangleq  x - \xhat$. Using \eqref{eq:plant} and \eqref{eq:ct_observer}, the estimation error dynamics is given by
\begin{equation}
\dot{\tilx}(t) = A_c \,\tilx(t) + d(t), \qquad A_c \triangleq A - L C \label{eq:error_dyn}
\end{equation}
Since Assumption~\ref{ass:detect} holds, the gain $L$ in \eqref{eq:error_dyn} is chosen such that $A_c$ is Hurwitz. Consequently, for any symmetric matrix $Q = Q^\top \succ 0$, there exists a unique symmetric matrix $P = P^\top \succ 0$ satisfying the Lyapunov equation \cite{antsaklis2007linear}:
\begin{equation}\label{eq:lyap}
A_c^\top P + P A_c = -Q.
\end{equation}

\begin{lemma}[ISS of the observer error dynamics]\label{lemma1}
Let $Q = Q^\top \succ 0$ be given and let $P = P^\top \succ 0$ satisfy \eqref{eq:lyap}.
Under Assumption~\ref{ass:rate}, the error dynamics \eqref{eq:error_dyn} is input-to-state stable (ISS) with respect to $d(t)$ in the sense that, for all $t \ge t_0$
\begin{equation}\label{eq:iss_bound_P}
V(t) \le e^{-\alpha (t-t_0)} V(t_0) + \frac{\beta}{\alpha}\,\bar d^{\,2}
\end{equation}
where $V(t) \triangleq \tilx(t)^\top P \tilx(t)$ and the constants
\[
\alpha \triangleq \frac{\lambda_{\min}(Q)}{2\,\lambda_{\max}(P)}, \qquad
\beta \triangleq \frac{2\,\|P\|^2}{\lambda_{\min}(Q)}
\]
are positive.
Consequently,
\begin{equation}\label{eq:iss_bound_x}
\|\tilx(t)\| \le
\sqrt{\frac{\lambda_{\max}(P)}{\lambda_{\min}(P)}}\,e^{-\frac{\alpha}{2}(t-t_0)}\|\tilx(t_0)\|
+
\sqrt{\frac{\beta}{\alpha\,\lambda_{\min}(P)}}\,\bar d
\end{equation}
\end{lemma}

\begin{proof}
The time derivative of $V$ along solutions of \eqref{eq:error_dyn} is:
\begin{align*}
\dot V &= \tilx^\top (A_c^\top P + P A_c)\tilx + 2\tilx^\top P d = -\tilx^\top Q \tilx + 2\tilx^\top P d
\end{align*}

\noindent Using Cauchy-Schwarz and Young inequalities yields
\begin{align*}
&\dot V \leq -\lambda_{\min}(Q)\|\tilx\|^2 + 2\,\|\tilx\|\,\|P\|\, \|d\| \cr
&\leq -\lambda_{\min}(Q)\|\tilx\|^2 + \frac{\lambda_{\min}(Q)}{2}\|\tilx\|^2 + \frac{2\|P\|^2}{\lambda_{\min}(Q)}\|d\|^2\cr
& \le -\frac{\lambda_{\min}(Q)}{2}\|\tilx\|^2 + \frac{2\|P\|^2}{\lambda_{\min}(Q)}\|d\|^2
\end{align*}

\noindent Since $V = \tilx^\top P \tilx \le \lambda_{\max}(P)\|\tilx\|^2$, it follows that $\|\tilx\|^2 \ge V/\lambda_{\max}(P)$, and using the bound on $d$ from Assumption \ref{ass:rate} in the above inequality, we get
\begin{equation}\label{eq:V_dissipation}
\dot V \le -\alpha V + \beta \bar d^{\,2}
\end{equation}

\noindent Consider the scalar comparison system
\[
\dot z(t) = -\alpha z(t) + \beta \bar d^{\,2}, \qquad z(t_0)=V(t_0)
\]
Its solution is given by
\[
z(t) = e^{-\alpha (t-t_0)} V(t_0)
+ \int_{t_0}^{t} e^{-\alpha (t-\tau)} \beta \bar d^{\,2}\, d\tau
\]
Evaluating the integral gives
\begin{align*}
z(t)
&= e^{-\alpha (t-t_0)} V(t_0)
+ (\beta/\alpha)\bar d^{\,2}\bigl(1-e^{-\alpha (t-t_0)}\bigr)\cr
& \le
e^{-\alpha (t-t_0)} V(t_0)
+ (\beta/\alpha)\bar d^{\,2}
\end{align*}
Since $\dot V \le -\alpha V + \beta \bar d^{\,2}$ and $z$ solves $\dot z = -\alpha z + \beta \bar d^{\,2}$ with $z(t_0)=V(t_0)$,
the comparison lemma \cite[Lemma~3.4]{khalil2002nonlinear} implies $V(t)\le z(t)$ for all $t\ge t_0$, which establishes \eqref{eq:iss_bound_P}. Finally, since $P=P^\top \succ 0$, $\lambda_{\min}(P)\|\tilx(t)\|^2 \le V(t)$ and $V(t_0)\le \lambda_{\max}(P)\|\tilx(t_0)\|^2$. Substituting these bounds in \eqref{eq:iss_bound_P} and then using $\sqrt{a+b}\le \sqrt{a}+\sqrt{b}$ for $a,b\ge 0$ yields \eqref{eq:iss_bound_x}.
\end{proof}

\begin{corollary}\label{cor:ct_uniform}
Under Assumptions~\ref{ass:rate} and \ref{ass:init_set}, for all $t \ge t_0$, the estimation error of the continuous-time observer \eqref{eq:ct_observer}-\eqref{eq:error_dyn} satisfies
\begin{equation}\label{eq:uniform_V}
V(t) \le V(t_0) + (\beta/\alpha)\,\bar d^{\,2}
\end{equation}
where $V(t)=\tilx(t)^\top P\tilx(t)$ and $\alpha,\beta$ are as in Lemma~\ref{lemma1}.
Consequently,
\begin{align}\label{eq:uniform_x}
\|\tilx(t)\| &\le
\sqrt{\frac{\lambda_{\max}(P)}{\lambda_{\min}(P)}}\,\|\tilx(t_0)\|
+
\sqrt{\frac{\beta}{\alpha\,\lambda_{\min}(P)}}\,\bar d \le M\cr
M &\triangleq
\sqrt{\frac{\lambda_{\max}(P)}{\lambda_{\min}(P)}}\,\bar x_0 \
+
\sqrt{\frac{\beta}{\alpha\,\lambda_{\min}(P)}}\,\bar d
\end{align}
\end{corollary}

\begin{proof}
Using $e^{-\alpha(t-t_0)}\le 1$  in \eqref{eq:iss_bound_P} yields \eqref{eq:uniform_V}. Subsequently, \eqref{eq:uniform_x} follows using $\lambda_{\min}(P)\|\tilx\|^2 \le V \le \lambda_{\max}(P)\|\tilx\|^2$, together with Assumption~\ref{ass:init_set}.
\end{proof}

\noindent Note that Lemma~\ref{lemma1} and Corollary~\ref{cor:ct_uniform} provide standard bounds for the continuous-time observers. These bounds will later serve as a baseline for comparison with the proposed observer.

\subsection{Motivation for a Hybrid Observer}

Lemma \ref{lemma1} establishes that the observer in \eqref{eq:ct_observer} is ISS with respect to bounded disturbances \cite{khalil2002nonlinear}, ensuring bounded estimation error. However, this does not guarantee that the residual (output error)
\begin{equation}\label{eq:residual}
r(t) \triangleq y(t) - C \xhat(t)
\end{equation}

\noindent remains within a prescribed envelope. In particular, disturbances may induce large transient residual deviations, even when the observer is stable. In observer-based control and supervisory applications, such deviations may be undesirable when state estimates are used for feedback or decision making. This motivates the following problem formulation.

\noindent \textbf{Problem statement:} Given any $\delta > 0$, and an initial residual satisfying $\|r(t_0)\|\le\delta$, design an observer such that
\[
\|r(t)\| \le \delta, \quad \forall t \ge t_0
\]
for all disturbances satisfying Assumption \ref{ass:rate}.

To this end, we use $r(t)$ as a supervisory signal and introduce a hybrid observer with continuous flows and discrete resets \cite{goebel2012hybrid}. This hybrid observer preserves the continuous-time robustness of the baseline Luenberger observer while enabling explicit containment of the residual norm.

\section{Observer with Prescribed Residual Bound}\label{sec:main}

\subsection{Hybrid Observer Formulation}\label{sec:design}
We augment the continuous-time observer \eqref{eq:ct_observer} with resets driven by the residual $r(t)$, defined in \eqref{eq:residual}. Let $\delta>0$ denote a designer-specified residual bound and let $\rho \in [0,1)$ denote a prescribed residual contraction factor. We define the resulting hybrid observer-plant dynamics as follows:
\begin{subequations}\label{eq:hyb_plant_observer}
\begin{align}
\dot{x} &= A x + B u + d
&& (x,\xhat) \in \mathcal{C} \label{eq:hyb_observer_specific}\\
\dot{\xhat} &= A \xhat + B u + L\bigl(y - C \xhat\bigr)
&& (x,\xhat) \in \mathcal{C} \label{eq:hyb_flow}\\
x^+ &= x^-
&& (x,\xhat) \in \mathcal{D} \label{eq:hyb_plant_jump}\\
\xhat^+ &= \xhat^- + J r^-
&& (x,\xhat) \in \mathcal{D} \label{eq:hyb_jump}
\end{align}
\end{subequations}

\noindent where $(x^-,\hat x^-)$ and $(x^+,\hat x^+)$ denote the plant and observer states immediately before and after a reset, respectively, $J \in \R^{n\times p}$ denotes a constant reset gain to be designed, and $r^-$ denotes the residual immediately before the reset. Note that the plant state $x$ remains continuous across resets, whereas only the observer state $\xhat$ is subject to jumps. The flow and jump sets for \eqref{eq:hyb_plant_observer} are defined as
\begin{equation}\label{eq:CD_sets}
\mathcal{C} \triangleq \{(x,\xhat):\|r\|\le\delta\}, \qquad
\mathcal{D} \triangleq \{(x,\xhat):\|r\|\ge\delta\}
\end{equation}

\noindent Note that at boundary points where both flow and jump are possible, jumps are given priority. Using $x^+=x^-$, together with $\tilx=x-\xhat$ and $r=C\tilx$ (which follows from \eqref{eq:plant_y} and \eqref{eq:residual}), the estimation-error and residual jump dynamics are
\begin{subequations}\label{eq:jump_error}
\begin{align}
\tilx^+ &= (\Id-JC)\tilx^- \label{eq:jump_error_states}\\
r^+ &= C\tilx^+ = C(\Id-JC)\tilx^- \label{eq:jump_error_res}
\end{align}
\end{subequations}

\noindent The properties of the hybrid observer are analyzed next.

%

\subsection{Analysis and Performance Guarantees}\label{sec:perf_guaratees}
The hybrid observer in \eqref{eq:hyb_flow}, \eqref{eq:hyb_jump} uses residual-triggered resets as a supervisory mechanism. The proposed observer architecture is constructed to enforce contraction of the residual at every reset, render a designer-prescribed residual envelope positively invariant, and preserve the robustness and boundedness properties of the baseline continuous-time observer in \eqref{eq:ct_observer}. The following theorem formalizes these guarantees and constitute the main result of the paper.

\begin{theorem}[Main result]\label{thm:main}
Suppose Assumptions \ref{ass:rate}-\ref{ass:fullrankC} hold. Let $\delta>0$ and $\rho \in [0,1)$ be
given and choose the reset gain $J$ as
\begin{equation}\label{eq:J_constructive}
J \triangleq (1-\rho)P^{-1}C^\top\bigl(CP^{-1}C^\top\bigr)^{-1}
\end{equation}
If $\|\res(t_0)\|\le\delta$, then the hybrid observer-plant system
\eqref{eq:hyb_plant_observer}-\eqref{eq:CD_sets} satisfies the following claims:
\begin{enumerate}
\renewcommand{\labelenumi}{(\roman{enumi})}
\item (\emph{Residual contraction at jumps})
\begin{equation}\label{eq:thm_res_contract}
\|\res^+\| = \rho\,\|\res^-\|
\end{equation}

\item (\emph{Forward invariance of the residual envelope}) The residual envelope, $\{r : \|r\| \le \delta\}$ is forward invariant, \emph{i.e.},
\begin{equation}\label{eq:thm_res_envelope}
\|r(t)\| \le \delta, \quad \forall t \ge t_0
\end{equation}
\vspace{-.08in}
\item (\emph{Strict Lyapunov decrease at jumps})
Let $V\triangleq\tilx^\top P\tilx$ and $S\triangleq CP^{-1}C^\top\succ0$. At every reset,
\begin{equation}\label{eq:nonexp_V}
V^+=V^--(1-\rho^2)r^{-\top}S^{-1}r^-
\end{equation}
In particular, since $\|r^-\|=\delta>0$,
\begin{equation}\label{eq: strict_decrease}
V^--V^+\ge\frac{(1-\rho^2)\delta^2}{\lambda_{\max}(S)}>0
\end{equation}

\item (\emph{Uniform boundedness of estimation error})
For all $t\ge t_0$,
\begin{equation}\label{eq:thm_hybrid_ub}
V(t) \le V(t_0) + \frac{\beta}{\alpha}\,\bar d^{\,2}
\end{equation}
where $\alpha,\beta$ are as in Lemma~\ref{lemma1}. Consequently,
\begin{align}\label{eq:thm_hybrid_ub_x}
\|\tilx(t)\|
&\le
M \\
M \triangleq \sqrt{\frac{\lambda_{\max}(P)}{\lambda_{\min}(P)}}\,\bar x_0
&+
\sqrt{\frac{\beta}{\alpha\,\lambda_{\min}(P)}}\,\bar d \notag
\end{align}
\end{enumerate}
\end{theorem}

\begin{proof}
We prove claims (i)-(iv) sequentially.

\noindent\emph{Proof of claim (i):}
Substituting  \eqref{eq:J_constructive} in \eqref{eq:jump_error}, we get
\begin{equation}\label{eq:res_plus}
\res^+ = C\Bigl(\Id - (1-\rho)P^{-1}C^\top(CP^{-1}C^\top)^{-1}C\Bigr)\tilx^-
\end{equation}

\noindent Since $C$ has full row rank from Assumption \ref{ass:fullrankC}, $CP^{-1}C^\top \succ 0$ and $CP^{-1}C^\top(CP^{-1}C^\top)^{-1}=\Id$. Using this in \eqref{eq:res_plus} and simplifying, we get
\[
\res^+ = \rho\,C\tilx^- = \rho\,\res^-
\]
which proves \eqref{eq:thm_res_contract} after taking norms on both sides.

\noindent\emph{Proof of claim (ii):}
During flows, the residual evolves continuously. By \eqref{eq:thm_res_contract}, immediately after any jump,
\[
\|\res^+\|=\rho\,\delta < \delta
\]
since $\rho \in [0,1)$. By definition of the flow and jump sets in \eqref{eq:CD_sets}, a jump is enforced whenever $\|\res\|=\delta$. Since $\|\res(t_0)\|\le \delta$ by assumption and $\res(t)$ is continuous during flows, the residual cannot exit the residual envelope $\{\res:\|\res\|\le \delta\}$. Therefore, this set is forward invariant, which establishes \eqref{eq:thm_res_envelope}.

\noindent\emph{Proof of claim (iii):}
Define the matrix
\begin{equation}\label{eq:projector}
\Pi \triangleq P^{-1}C^\top(CP^{-1}C^\top)^{-1}C
\end{equation}

\noindent Using  \eqref{eq:projector} and \eqref{eq:J_constructive}, we can express
\[
\Id - JC = \Id - (1-\rho)\Pi = (\Id-\Pi)+\rho\Pi
\]
Hence, from \eqref{eq:jump_error_states},
\[
\tilx^+ = (\Id-\Pi)\tilx^- + \rho \Pi \tilx^-
\]

Let $a \triangleq (\Id-\Pi)\tilx^-$ and $b \triangleq \Pi\tilx^-$. Then $\tilx^- = a+b$ and $\tilx^+ = a + \rho b$.
Consider the Lyapunov function $V(\tilx)=\tilx^\top P\tilx$, which defines the $P$-weighted metric associated with the estimation error. Evaluating $V^+$ gives
\begin{equation}\label{eq:Vplus_expand}
V^+ = (a+\rho b)^\top P(a+\rho b)
= a^\top P a + 2\rho\, a^\top P b + \rho^2 b^\top P b
\end{equation}

\noindent We next show that the cross term in \eqref{eq:Vplus_expand} vanishes. Note that
\begin{equation}\label{eq:CPi_eq_C}
C\Pi = CP^{-1}C^\top(CP^{-1}C^\top)^{-1}C = C
\end{equation}
and therefore $C(\Id-\Pi)=0$. Using $b=\Pi\tilx^-$ and $a=(\Id-\Pi)\tilx^-$, we have
\[
a^\top P b
= \tilx^{-\,\top}(\Id-\Pi)^\top P\,\Pi \tilx^-
\]

\noindent Using $P\Pi = C^\top(CP^{-1}C^\top)^{-1}C$ and $C(\Id-\Pi)=0$,
\begin{align*}
(\Id-\Pi)^\top P\,\Pi
&= (\Id-\Pi)^\top C^\top(CP^{-1}C^\top)^{-1}C \cr
&= \big(C(\Id-\Pi)\big)^\top (CP^{-1}C^\top)^{-1}C
= 0
\end{align*}
which implies $a^\top P b = 0$. Substituting $a^\top P b = 0$ in \eqref{eq:Vplus_expand},
\[
V^+ = a^\top P a + \rho^2 b^\top P b \le a^\top P a + b^\top P b
\]
Finally, since $\tilx^- = a+b$ and $a^\top P b=0$, we have
\[
V^- = (a+b)^\top P(a+b) = a^\top P a + b^\top P b
\]

\noindent Computing $V^- - V^+$, and using $b=P^{-1}C^\top S^{-1}r^-$ gives
\[
V^--V^+=(1-\rho^2)b^\top Pb=(1-\rho^2)r^{-\top}S^{-1}r^-
\]
This proves \eqref{eq:nonexp_V}. Since $S\succ0$ and $\|r^-\|=\delta>0$, the strict decrease bound in \eqref{eq: strict_decrease} follows.

\noindent\emph{Proof of claim (iv):}
Let $t_0$ denote the initial time and let $t_1,t_2,\ldots$ denote the subsequent reset times of the hybrid observer. Let $V(t)=\tilx(t)^\top P\tilx(t)$ and let $z(t)$ be the comparison function in the proof of Lemma \ref{lemma1}, \emph{i.e.}, the unique solution of
\begin{equation}\label{eq:z_dyn_claimiv}
\dot z(t) = -\alpha z(t) + \beta \bar d^{\,2}, \qquad z(t_0)=V(t_0)
\end{equation}

\noindent We prove by induction on $k\in\mathbb{N}_0$ that
\begin{equation}\label{eq:V_le_z_induction}
V(t) \le z(t), \qquad \forall t \in [t_k,t_{k+1})
\end{equation}
where $t_{k+1}\triangleq \inf\{t>t_k:\,(x(t),\xhat(t))\in\mathcal{D}\}$ if another jump occurs, and $t_{k+1}=\infty$ otherwise.

\emph{Base case ($k=0$):} If no reset occurs at $t_0$, then $V(t_0)=z(t_0)$. If $\|r(t_0)\|=\delta$, jump priority gives an initial reset and claim (iii) yields $V(t_0^+)\le V(t_0)=z(t_0)$. Thus, on the first flow interval $[t_0,t_1)$, the comparison argument in Lemma~\ref{lemma1} implies $V(t)\le z(t)$.

\emph{Inductive step:} Now, suppose \eqref{eq:V_le_z_induction} holds on $[t_k,t_{k+1})$ for some $k\ge 0$.
Since $z(t)$ is continuous, taking the left limit in \eqref{eq:V_le_z_induction} yields $V(t_{k+1}^-)\le z(t_{k+1})$ . At the jump time $t_{k+1}$, claim (iii) yields $V(t_{k+1}^+) \le V(t_{k+1}^-)$. Thus,
\[
V(t_{k+1}^+) \le z(t_{k+1})
\]

\noindent On the subsequent flow interval $[t_{k+1},t_{k+2})$, $\tilde x$ again satisfies \eqref{eq:error_dyn}, and hence the differential inequality \eqref{eq:V_dissipation} holds for almost all $t\in[t_{k+1},t_{k+2})$. Since $z(t)$ satisfies \eqref{eq:z_dyn_claimiv}, and the scalar comparison system is order-preserving with respect to its initial value, the inequality
\[
V(t) \le z(t), \qquad \forall t\in[t_{k+1},t_{k+2})
\]
follows from the initial condition $V(t_{k+1}^+) \le z(t_{k+1})$. This completes the induction argument and thus $V(t)\le z(t)$ over the solution domain. The minimum dwell-time result established in Section \ref{sec:wellposed} excludes Zeno behavior, and hence the solution is defined for all $t\ge t_0$. Using the explicit expression for $z(t)$ given in the proof of Lemma~\ref{lemma1} yields
\[
V(t) \le e^{-\alpha (t-t_0)} V(t_0) + (\beta / \alpha)\,\bar d^{\,2},
\qquad \forall t\ge t_0
\]
Using $e^{-\alpha(t-t_0)}\le 1$ establishes \eqref{eq:thm_hybrid_ub}. Finally, the bound \eqref{eq:thm_hybrid_ub_x} follows by using $\lambda_{\min}(P)\|\tilx\|^2 \le V \le \lambda_{\max}(P)\|\tilx\|^2$, together with Assumption~\ref{ass:init_set}.
\end{proof}

\begin{figure}[t!]
 \centering
\psfrag{A}[c][c][1][0]{\scriptsize $t_0$}
\psfrag{G}[c][c][1][0]{\scriptsize $t_1^-$}
\psfrag{B}[c][c][1][0]{\scriptsize $t_1^+$}
\psfrag{C}[c][c][1][0]{\scriptsize $ t_2^-$ }
\psfrag{D}[c][c][1][0]{\scriptsize $t_2^+$}
\psfrag{Q}[c][c][1][0]{\scriptsize  $t_3^-$}
\psfrag{F}[c][c][1][0]{\scriptsize  $t_3^+$}
\psfrag{H}[c][c][1][0]{\scriptsize  $\rho \delta$}
\psfrag{I}[c][c][1][0]{\scriptsize  $\delta$}
 \includegraphics[width=0.57\linewidth]{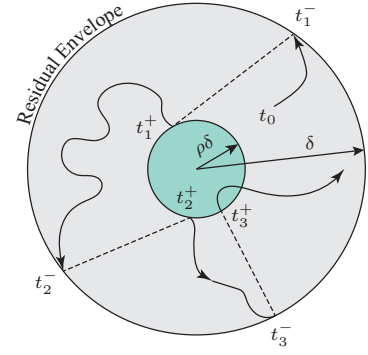}

\caption{Geometric interpretation of residual contraction and forward invariance of the prescribed residual envelope under the hybrid plant-observer \eqref{eq:hyb_plant_observer}-\eqref{eq:CD_sets}.}
 \label{fig3}
  \vspace{-0.2in}
 \end{figure}
 
The contraction of the residual and the forward invariance of the residual envelope can be visualized using Fig.~\ref{fig3}. The outer ball represents the residual envelope $B = \{\res : \|\res\| \le \delta\}$. The residual at the initial time $t_0$ lies within this bound. At instants $t_1^-$, $t_2^-$, and $t_3^-$, when the residual reaches the boundary $\|\res\| = \delta$, a reset is triggered. The reset contracts the residual to the boundary of the inner ball whose radius is $\rho \delta$. As a result, the residual never leaves the set $B$, and the residual envelope remains forward invariant.
 
\begin{remark}
Claim (iv) shows that the hybrid observer provides the same uniform bound for the estimation error as the continuous time observer in Corollary \ref{cor:ct_uniform}.
\end{remark}


\begin{remark}[Measurement noise]
If the output is corrupted by bounded noise that is continuous at reset instants, the reset still contracts the residual. However, $\delta$ should be chosen sufficiently above the noise level to avoid noise-driven switching.
\end{remark}

\begin{corollary}\label{corr:initialization_outside}[Initialization outside the residual envelope]
Suppose Assumptions~\ref{ass:rate}-\ref{ass:fullrankC} hold. If the reset gain \eqref{eq:J_constructive} is chosen with $\rho=0$ and $\|\res(t_0)\|\ge \delta$, then the residual satisfies:
\[
\|\res(t)\|\le \delta, \qquad \forall t\ge t_0^+
\]
\end{corollary}


\begin{proof}
With  $\rho=0$ in \eqref{eq:thm_res_contract}, $\res^+ =0$. Hence, after the first reset at $t_0$, the residual lies in the interior of the residual envelope $\{r:\|r\|\le \delta\}$. The forward-invariance argument of Theorem \ref{thm:main} then implies $\|r(t)\|\le\delta$ for all $t\ge t_0^+$.
\end{proof}

\begin{remark}
If $\delta \ge \|C\|\bar x_0$, then $\|r(t_0)\|\le\delta$ follows from Assumption~\ref{ass:init_set}. If instead $\|r(t_0)\|>\delta$, choosing $\rho=0$ gives $r(t_0^+)=0$, and invariance thereafter follows from Corollary~\ref{corr:initialization_outside}.
\end{remark}

\subsection{Exclusion of Zeno Behavior - Minimum Dwell Time}\label{sec:wellposed}

To exclude Zeno behavior, we establish that two consecutive resets are separated by a strictly positive minimum dwell time \cite{goebel2012hybrid}. We now establish such a minimum dwell time for the observer \eqref{eq:hyb_plant_observer}. Since jumps are triggered when $\|r\|=\delta$, it suffices to bound the growth rate of the residual during flows.
During flows, $r(t) = C\tilx(t)$, and therefore using \eqref{eq:error_dyn}
\[
\dot r(t)
=
C\dot{\tilx}(t)
=
C\left[A_c \tilx(t) + d(t)\right]
\]

\noindent Using the triangle inequality and Assumption \ref{ass:rate}, we get
\begin{equation}\label{eq:r_dot_bound}
\|\dot r(t)\|
\le
\|C\|\,\|A_c\|\,\|\tilx(t)\|
+
\|C\|\,\bar d
\end{equation}

\noindent Using \eqref{eq:thm_hybrid_ub_x} in \eqref{eq:r_dot_bound} yields
\begin{equation}\label{eq:residual_rate}
\|\dot r(t)\| \le \bar \ell,
\qquad
\bar \ell
\triangleq
\|C\|\,\|A_c\|\,M
+
\|C\|\,\bar d
\end{equation}

\noindent Let $t_k$ and $t_{k+1}$ denote two consecutive reset times. Immediately after $t_k$, $\|r(t_k^+)\|=\rho\delta$. A subsequent reset requires $\|r\|$ to increase to $\delta$. Integrating \eqref{eq:residual_rate} and taking norms, we obtain
\[
\|r(t)\|
\le
\|r(t_k^+)\|
+
\int_{t_k}^{t} \|\dot r(s)\| ds
\le
\rho\delta + \bar\ell (t-t_k)
\]
Thus, the earliest time at which $\|r(t)\|$ can reach $\delta$
satisfies
$
\delta
\le
\rho\delta + \bar\ell (t_{k+1}-t_k)
$, 
which implies
\begin{equation}\label{eq:dwell_time}
t_{k+1}-t_k
\ge
\left[(1-\rho)\delta\right]/ \bar\ell  >0
\end{equation}

\noindent Therefore, a strictly positive minimum dwell time is guaranteed between consecutive jumps.

\section{Results}\label{sec:results}

We compare the performance of the proposed observer \eqref{eq:hyb_plant_observer}-\eqref{eq:CD_sets} with that of a standard Luenberger observer in \eqref{eq:ct_observer} using simulations. The system matrices in \eqref{eq:plant} are chosen as:
\begin{align}
\label{eq:sim_plant_parameters}
A = \left[\begin{array}{rrrr}-0.5 0& 1.00 & 0.00 &1.00\cr -1.00 & -0.25  & 0.60 & -0.50 \cr 0.00& -0.80 &-0.15 &0.70\cr 1.00 & 0.00 & 0.50 & 0.00\end{array}\right] &\,\, B = \begin{bmatrix}0\cr0\cr1\cr0\end{bmatrix}\cr
C = \begin{bmatrix} 1 & 0 & 0 & 0\cr 0 & 0 & 1 & 0 \end{bmatrix}
\end{align}

\noindent The plant is subjected to a time-varying bounded disturbance $d(t)= \begin{bmatrix} d_1(t) & 0 & 0 & 0\end{bmatrix}^\top$ such that the disturbance is active only on the interval $t\in[2,5]$, described as follows
\begin{equation}\label{eq:disturbance_parameters}
d_1(t) =
\begin{cases}
0, & t < 2 \\
15 \sin(2t), & 2 \le t < 5 \\
0, & t \ge 5
\end{cases}
\end{equation}
\begin{figure}[t!]
 \centering
\psfrag{A}[c][c][1][0]{\scriptsize {time (sec)}}
\psfrag{B}[c][c][1][0]{\scriptsize $\|r(t)\|$}
\psfrag{C}[c][c][1][0]{\scriptsize $ (1- \rho) \delta$ }
\psfrag{D}[c][c][1][0]{\scriptsize $\delta$}
\psfrag{E}[c][c][1][0]{\scriptsize  $d_1(t)$}
 \includegraphics[width=\linewidth]{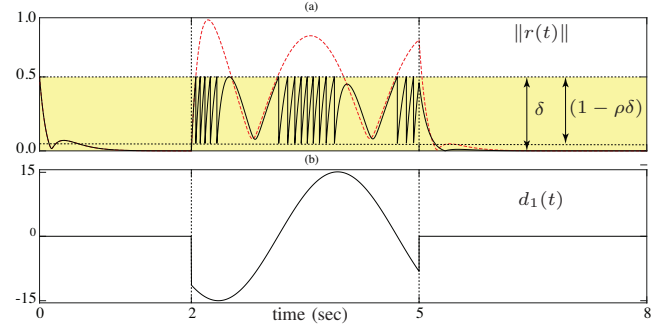}

\caption{(a) Output-residual norm $\|r(t)\|$ for the continuous-time observer  (dashed red) and the proposed hybrid observer (solid black).  The prescribed residual bound $\delta$ is indicated by the horizontal line. The hybrid observer enforces $\|r(t)\|\le\delta$ via residual-triggered resets, whereas the continuous-time observer exceeds the bound when the disturbance is active. (b) Disturbance signal $d_1(t)$, active on the interval $2 \le t \le 5$ s.}
 \label{fig1}
 \end{figure}
   \begin{figure}[b!]
 \centering
\psfrag{A}[c][c][1][0]{\scriptsize {$\tilx_1(t)$}}
\psfrag{B}[c][c][1][0]{\scriptsize {$\tilx_2(t)$}}
\psfrag{C}[c][c][1][0]{\scriptsize {$\tilx_3(t)$}}
\psfrag{D}[c][c][1][0]{\scriptsize {$\tilx_4(t)$}}
\psfrag{E}[c][c][1][0]{\scriptsize {time (sec)}}
 \includegraphics[width=\linewidth]{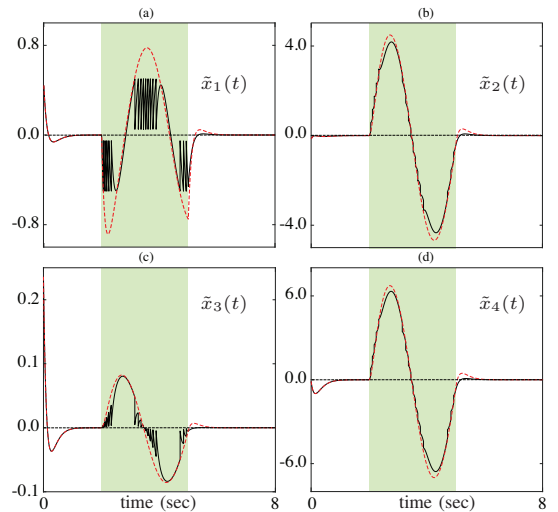}

\caption{Estimation errors $\tilde x_i(t)$, $i=1,\dots,4$, for the  continuous-time observer (dashed red) and the hybrid observer (solid black). The shaded region corresponds to the interval during which  the disturbance is active.}
 \label{fig2}

 \end{figure}

\noindent It can be verified that Assumptions~\ref{ass:rate}-\ref{ass:detect} are satisfied. The observer gain $L$ is designed by pole placement such that the eigenvalues of $A_c$ in \eqref{eq:error_dyn}  are located at $-5$, $-6$, $-7$, and $-8$. The initial conditions ($t_0 = 0$) are chosen as
\[
x(0) = \begin{bmatrix} 1.5 & -0.5 & 0.8 & -0.4 \end{bmatrix}^\top,
\qquad
\hat x(0) = 0.7\,x(0)
\]
so that a nonzero initial estimation error is present. The residual bound $\delta$ in \eqref{eq:CD_sets}
is selected as $\delta = \|r(0)\| = \|C(x(0)-\hat x(0))\|$. The reset gain of the hybrid observer in \eqref{eq:J_constructive} is constructed with $\rho = 0.1$, corresponding to a $90\%$ contraction of the residual magnitude at each reset, and the matrix $P$ was obtained by choosing $Q$ as identity matrix in \eqref{eq:lyap}. Simulation results are presented in Figs.~\ref{fig1} and~\ref{fig2}.

Fig.~\ref{fig1}(a) compares the residual norm $\|r(t)\|$ for the  continuous-time observer (dashed red) and the proposed hybrid observer (solid black).  Fig.~\ref{fig1}(b) shows the disturbance signal $d_1(t)$. As seen in Fig.~\ref{fig1}(b), the disturbance is active on the interval  $2 \le t \le 5$~s. During this interval, the residual of the continuous-time  observer exceeds the prescribed bound $\delta$, as shown in Fig.~\ref{fig1}(a). In contrast, the hybrid observer enforces the residual constraint  $\|r(t)\| \le \delta$ for all $t \ge 0$, in agreement with Theorem~\ref{thm:main}. 
Whenever the residual norm reaches the boundary $\|r(t)\|=\delta$,  a reset is triggered and the residual is instantaneously contracted to  $\rho\delta$, consistent with the construction of the reset gain in  \eqref{eq:J_constructive}. This contraction behavior is clearly visible at each reset event in Fig.~\ref{fig1}(a).

Figs.~\ref{fig2}(a)-(d) show the individual state-estimation errors  $\tilde x_1,\tilde x_2,\tilde x_3,$ and $\tilde x_4$, respectively.  Dashed red lines correspond to the continuous-time observer, while solid black curves correspond to the hybrid observer. When the disturbance becomes active (shaded interval), the observer resets introduce jumps in the state-estimation errors. In particular, the magnitude of $\tilde x_1$ decreases after each reset, while the remaining estimation errors also exhibit reductions at several reset instants and remain close to those of the continuous-time observer.

It should be noted that the analysis guarantees a decrease of the Lyapunov function $V=\tilde x^\top P\tilde x$ at each reset, rather than a decrease of every individual estimation-error component. Since the reset gain $J$ in \eqref{eq:J_constructive} depends on $P$, which in turn depends on $Q$ in \eqref{eq:lyap}, the choice of $Q$ influences how the reset correction is distributed among the individual estimation-error components. Systematic selection of $Q$ to shape this state-wise reset behavior is left for future work.

\begin{figure}[t!]
 \centering
\psfrag{A}[c][c][1][0]{\scriptsize {time (sec)}}
\psfrag{B}[c][c][1][0]{\scriptsize $\|r(t)\|$}
\psfrag{C}[c][c][1][0]{\scriptsize $ 0.9 \delta$ }
\psfrag{D}[c][c][1][0]{\scriptsize $\delta$}
\psfrag{E}[c][c][1][0]{\scriptsize  $V(t) = \tilx(t)^\top\! P \tilx(t)$}
 \includegraphics[width=0.9\linewidth]{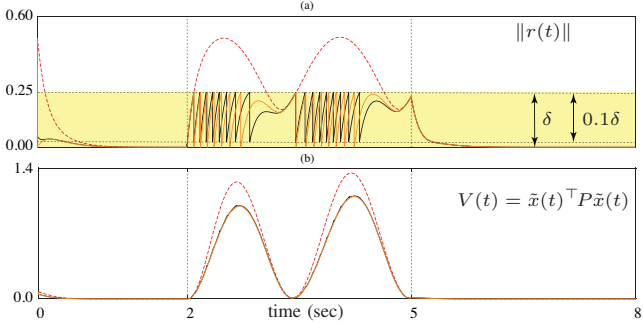}
 \vspace{-0.1in}
\caption{(a) Plot of $\|r(t)\|$ for the continuous-time observer  (dashed red) and the proposed hybrid observer with $\|r(0)\| >\delta$, $\rho = 0$ (solid orange), and with $\|r(0)\| < \delta$, $\rho = 0.1$ (solid black).  (b) $V(t)$ with time.}
\vspace{-0.2in}
 \label{fig4}
 \end{figure}

We also evaluated our design for a detectable but not observable pair $(A,C)$ by choosing the rows of $A$ as$[-0.50 \;\;1.00 \;\; 0.00 \;\; 0.00]$, $[-1.00\;\;-0.25\;\;0.60 \;\;0.00]$, $[0.00\;\; -0.80\;\; -0.15\;\; 0.00]$ and $[0.00\;\; 0.00\;\; 0.00\;\; -2.00]$, rendering the fourth state unobservable while preserving detectability. The observable poles were placed at $-5$, $-6$, and $-7$. Two cases were considered: (i) $\|r(0)\|>\delta$ with $\rho=0$ and (ii) $\|r(0)\|<\delta$ with $\rho=0.1$. As shown in Fig.~\ref{fig4} (a), a reset at $t=0$ in case (i) enforces invariance thereafter, while case (ii) remains invariant throughout. The plot of $V$ with $t$ in Fig.~\ref{fig4}(b) confirms bounded estimation error.

\section{Conclusion}

This paper proposed a reset observer that enforces a prescribed bound on the observer residual (output error). The observer combines a continuous-time Luenberger observer with state resets triggered when the residual norm reaches a specified threshold. The reset structure guarantees residual contraction at jump instants, forward invariance of the prescribed residual envelope, and decrease of the estimation-error Lyapunov function, while preserving the uniform boundedness properties of the continuous-time observer. Simulation results support the analysis and show that the residual remains within the prescribed bound under bounded disturbances, whereas a standard Luenberger observer with the same gains may violate this bound.

In residual-based fault detection, unusually frequent resets may provide an additional indication of deviation from nominal behavior. Future work will investigate how reset frequency and reset signatures can be used to develop diagnostic strategies, particularly for distinguishing persistent deviations from transient disturbances. Other directions include robustness under measurement noise and noise-aware selection of the residual bound, integration with output-feedback control via separation-type analysis, and extension of the framework to nonlinear and sampled-data systems.

\bibliographystyle{IEEEtran}
\bibliography{ref}

\end{document}